%% file: hal.tex
\documentclass[conference]{IEEEtran}
\usepackage{proof}
\usepackage{xspace}
\usepackage{latexsym}
\usepackage{amsmath}
\usepackage{amssymb}
\usepackage{url}
\usepackage{stmaryrd}
\usepackage{cmll} 

\newtheorem{prop}{Proposition}
\input macros.tex

\newenvironment{proof}{\noindent {\emph{Proof}}~}{$\Box$}

\begin{document}
\title{A system of inference based on proof search:\\ 
       an extended abstract}
\author{\IEEEauthorblockN{Dale Miller}
\IEEEauthorblockA{
  \textit{Inria Saclay \& LIX, Institut Polytechnique de Paris}\\
  Palaiseau France \\
  Orcid: 0000-0003-0274-4954}}
\maketitle
\begin{abstract}
Gentzen designed his natural deduction proof system to ``come as close
as possible to actual reasoning.''  Indeed, natural deduction proofs
closely resemble the static structure of logical reasoning in
mathematical arguments.  However, different features of inference are
compelling to capture when one wants to support the process of
searching for proofs.  \framework (Proof Search Framework) attempts to capture
these features naturally and directly.  The design and metatheory of
\framework are presented, and its ability to specify a range of proof
systems for classical, intuitionistic, and linear logic is
illustrated.
\\[3pt] [What follows is a slightly revised version of the
paper that appears in the Proceedings of LICS 2023.]\\[1pt]
\end{abstract}


\section{Introduction}

Inference and proofs are often described using proof rules of various
shapes.  For example, natural deduction and sequent calculus use figures such as 
\[   \newcommand{\xx}{\underline{~}}
  \vcenter{\infer{-}{- & - & \deduce{-}{\deduce{\vdots}{(-)}}}}
  \quad\hbox{and}\quad
  \vcenter{
  \infer{\xx,~\xx\longrightarrow \xx,~\xx\strut}
        {\xx\longrightarrow\xx,~\xx \qquad \xx,~
         \xx\longrightarrow \xx\strut}}
\]
These figures, introduced by Gentzen~\cite{gentzen35}, rely on several
punctuation marks such as the horizontal bar (to separate premises
from conclusion), vertical dots (for reasoning from assumptions),
parenthesized formulas (for discharging a formula), and the sequent
arrow.  The logical force implicit in the punctuation marks used to
describe proofs can invade the logic specified in that framework.  As
Wittgenstein has stated: ``Signs for logical operations are
punctuation marks.'' (Tractatus 5.4611, 1922).  While such influence
of the framework might be hard to avoid in general, we should be aware
of its influence and, at times, look for alternative systems of
punctuation.

Gentzen declared that his natural deduction system NJ was ``a formal
system which comes as close as possible to actual
reasoning''~\cite{gentzen35}.  Indeed, his natural deduction proof
systems have had great success ranging from being used in the teaching
of logical reasoning to the formal encoding of proofs as dependently
typed $\lambda$-terms.  However, since Gentzen's introduction of such
notation some four score and eight years ago, many
different priorities for logic and proof have appeared.

While natural deduction and sequent calculus have been used
successfully to describe the static structure of complete proofs (and
their transformation via normalization and cut elimination), the
dynamic structure of the \emph{search} for proofs is less well
captured by his systems.  Here, issues such as \emph{partial proofs},
\emph{invertible inference rules}, and \emph{don't care} and \emph{don't
know} non-determinism are particularly important to support.  

\section{Design motivated}
\label{sec:motivated}

Consider a sheet of paper on which a mathematician has written several
formulas at the top and one at the bottom.  Such a sheet is
useful to represent a proof gap, where one needs to find a logical
argument that connects the \emph{given} formulas at the top to the
intended \emph{consequence} written at the bottom.  In \framework, the
search for a proof is encoded as the rewriting of a collection of such
proof gaps recorded on sheets.  A sheet might rewrite to no additional
sheets if it is recognized as trivially proved: for example, because
the formula at the bottom of the sheet is also present at the top.
On the other hand, 
a sheet can rewrite to other sheets if solving those
additional sheets is understood as a way to solving the originating
sheet.  For example, a sheet containing the formula $even(n)\vee
odd(n)$ at the top can be rewritten to make two identical copies
except that $even(n)$ is put at the top in one and $odd(n)$ is put at
the top into the other.  The rule of cases would justify such a rewriting.
\framework encodes such sheets as \emph{multisets} of
\emph{tagged} formulas: if the logical formula $B$ appears at the
top of the sheet, it is placed into that multiset as $\lft{B}$; if it
appears at the bottom, it is placed into that multiset as $\rght{B}$
(see Section~\ref{ssec:classical}).

A feature of inference rules that \framework puts in prominence is the
difference between \emph{multiplicative} and \emph{additive} inference
rules.  The following are examples of the additive and multiplicative
versions of the right introduction for conjunction.
\[
  \infer{\twoseq{\Gamma}{A\wedge B,\Delta}}
        {\twoseq{\Gamma}{A,\Delta} & \twoseq{\Gamma}{B,\Delta}}
  \qquad
  \infer{\twoseq{\Gamma_1,\Gamma_2}{A\wedge B,\Delta_1,\Delta_2}}
        {\twoseq{\Gamma_1}{A,\Delta_1} & 
         \twoseq{\Gamma_2}{B,\Delta_2}}
\]
More generally, an inference rule is \emph{additive} if every
side-formula occurrence (\ie, those in $\Gamma$ and $\Delta$) also
occur in every premise.  A rule is called \emph{multiplicative} if
every side-formula occurrence (\ie, those in $\Gamma_1$, $\Gamma_2$,
$\Delta_1$, and $\Delta_2$) also occurs in exactly one premise.  A
rule with exactly one premise is additive exactly when it is 
multiplicative.  \framework contains two operators $\Plus$ and
$\Times$ responsible for injecting additive and multiplicative
features into inference systems encoded into it.

The multiplicative features of \framework are easily illustrated by
the need to rewrite multisets to other multisets.  In particular,
multisets will be encoded as expressions built from (some fixed set
of) atomic expressions along with $\Times$ for building a non-empty
multiset and its unit $\One$ denoting an empty multiset.  For example,
if $a$, $b$, and $c$ are atomic expressions, then $a\Times a\Times b$
denotes the multiset that contains two occurrences of $a$ and one
occurrence of $b$.  Rewriting a multiset $M$ to
another multiset $N$ using the rule $M_1\mapsto M_2$ (where $M_1$ and
$M_2$ are also multisets) is done using the following steps. (1) Split
$M$ into two parts $M'$ and $M''$.  (2) Determine that $M'$ is the
same multiset as $M_1$. (3) Identify $N$ with the multiset union of
$M_2$ and $M''$.  The following small proof system (extended
in the next section) can be used to describe such a computation.
\[
  \infer{\twoseq{}{\One,\Delta}}
        {\twoseq{}{\Delta}}
  \qquad
  \infer{\twoseq{}{E_1\Times E_2,\Delta}}
        {\twoseq{}{E_1,E_2,\Delta}}
  \qquad
  \infer{\twoseq{E}{E}}{}
\]
\[
  \infer{\twoseq{E_1\mapsto E_2}{\Delta_1,\Delta_2}}
        {\twoseq{E_1}{\Delta_1} & 
         \twoseq{}{E_2, \Delta_2}}
  \quad
  \infer{\twoseq{\One}{}}{}
  \quad~
    \infer{\twoseq{E_1\Times E_2}{\Delta_1, \Delta_2}}
        {\twoseq{E_1}{\Delta_1} & 
         \twoseq{E_2}{\Delta_2}}
\]
The left-introduction rule for $\mapsto$ achieves the three steps
mentioned above.  Step (1) is captured by splitting a multiset into
the union of $\Delta_1$ and $\Delta_2$ in that rule's conclusion.
Steps (2) and (3) are captured by the proofs of its left and right
premises, respectively.  The rewriting of the multiset $\{a, a, b\}$
into $\{a,c\}$ by the rule that replaces $a$ and $b$ with $c$ is
witnessed by a derivation of $\twoseq{a\Times b\mapsto c}{a, a, b}$
from the open premise $\twoseq{}{a,c}$.

Additive features are also incorporated into \framework using $\Plus$
and its unit $\Zero$: in particular, collections of multisets are
represented as a $\Plus$ of $\Times$ of atomic expressions.  Below we
list three additional features abstracted from searching for proofs
based on evolving collections of sheets.

\paragraph{Linear and classical realms}
When rewriting a sheet of paper to possibly other sheets, it is
usually the case that some items are retained while others might
disappear.  In particular, an assumption at the top of a sheet is
usually retained on all subproblems that are eventually rewritten from
it, while the goal formula on one sheet may or may not change.  For
example, if the goal formula is $A\iimp B$, then that goal is
\emph{replaced} by the goal formula $B$ with $A$ simultaneously added
at the top of the sheet.  The \framework recognizes this distinction
by classifying atomic expressions as being in either the \emph{linear
realm}---where such expressions might be deleted or replaced---or
the \emph{classical realm}---where such expressions persist through
all evolutions of a multiset.  (There is a strong influence of linear
logic~\cite{girard87tcs} on the design of \framework.)

\paragraph{Bottom-up and top-down reasoning}  These proof search 
styles appear in various different disguises in computational logic.
They differentiate Prolog from Datalog and tableaux from
resolution~\cite{chaudhuri08jar}.  Term representation is often
described using top-down proof structures, while term representations
that allow for explicit sharing can be justified using bottom-up
proof structures~\cite{miller23csl}.  In \framework, this distinction
comes into play using the notions of bias assignment and debts.

\paragraph{Don't care and don't know non-determinism} The
non-determinism encountered in the search for proofs can be
categorized as being either the \emph{don't care} or \emph{don't know}
varieties.  In \framework, inference rules will eventually be
organized into two phases: the \emph{right phase} will capture don't
care non-determinism and the \emph{left phase} will capture don't
know non-determinism.

\framework is presented using two inference system.  The basic
system, \basic, is presented in Section~\ref{sec:basic} while a more
structured variant, \twophase, is presented in Section~\ref{sec:two
  phase}.

\section{The basic inference system \basic}
\label{sec:basic}

\begin{figure}
\textsc{The right rules}
\[
  \infer{\twoseq{\Gamma}{\Zero,\Delta}}{}
  \qquad
  \infer{\twoseq{\Gamma}{E_1\Plus E_2,\Delta}}
        {\twoseq{\Gamma}{E_1,\Delta} &
         \twoseq{\Gamma}{E_2,\Delta}}
\]
\[
  \infer{\twoseq{\Gamma}{\One,\Delta}}
        {\twoseq{\Gamma}{\Delta}}
  \qquad
  \infer{\twoseq{\Gamma}{E_1\Times E_2,\Delta}}
        {\twoseq{\Gamma}{E_1,E_2,\Delta}}
\]
\textsc{The left rules}
\[
  \infer{\twoseq{\One}{}}{}
  \quad~
  \infer{\twoseq{\Gamma,R_1\Plus R_2}{\Delta}}
        {\twoseq{\Gamma,R_i}{\Delta}}
  \quad~
    \infer{\twoseq{\Gamma_1,\Gamma_2,R_1\Times R_2}{\Delta_1, \Delta_2}}
        {\twoseq{\Gamma_1,R_1}{\Delta_1} & 
         \twoseq{\Gamma_2,R_2}{\Delta_2}}
\]
\[
  \infer{\twoseq{\Gamma_1,\Gamma_2,R\mapsto E}{\Delta_1,\Delta_2}}
        {\twoseq{\Gamma_1,R}{\Delta_1} & 
         \twoseq{\Gamma_2}{E, \Delta_2}}
  \qquad
  \infer{\twoseq{\Gamma,R\Mapsto E}{\Upsilon,\Delta}}
        {\twoseq{R}{\Upsilon} & 
         \twoseq{\Gamma}{E,\Delta}}
\]
\[
  \infer[\decide,~\hbox{$\Gamma\subseteq\Rules$ is non-empty and finite}]
        {\twoseq{}{\Delta}}
        {\twoseq{\Gamma}{\Delta}}
\]
\[
  \infer[\debit_1]{\twoseq{\Gamma,A}{\Delta}}
                  {\twoseq{\Gamma}{\bar{A},\Delta}\quad\bias{A}=+1}
\]
\[
  \infer[\debit_2]{\twoseq{\Super}{\Upsilon}}
                  {\twoseq{}{\bar{\Super},\Upsilon}\quad\bias{\Super}=+2}
\]
\textsc{The identity rules}
\[
  \infer[\init]{\twoseq{E}{E}}{}
  \qquad
  \infer[\initd]{\twoseq{}{\bar{A},A}}{}
\]
\textsc{The structural rules}
\[
  \infer[\contrR]{\twoseq{\Gamma}{\Delta,\Super}}
                 {\twoseq{\Gamma}{\Delta,\Super,\Super}}
  \qquad
  \infer[\weakR]{\twoseq{\Gamma}{\Delta,\Super}}
                {\twoseq{\Gamma}{\Delta}}
\]
\caption{The \basic proof system.}
\label{fig:b3}
\end{figure}

\fref{b3} contains the inference system \basic, which contains all the
features we have motivated so far: additive and multiplicative
structures, proof state rewriting, debts, and the linear and classical
realms.  The schematic variables used in \fref{b3} are the following.
The variable $A$ ranges over some fixed set of atomic expressions.
The variables $E$ and $R$ range over expressions and rules and are
defined as follows.
\begin{align*}
E &::= A~\sep~\Zero~\sep~E_1\Plus E_2~\sep~\One~\sep~E_1\Times E_2\\
R &::= A~\sep~\Zero~\sep~R_1\Plus R_2~\sep~\One~\sep~R_1\Times R_2\\
     & \qquad\quad\sep~ R\mapsto E ~\sep~ R\Mapsto E
\end{align*}
The operators $\mapsto$ and $\Mapsto$ associate to the left while the
operators $\Plus$ and $\Times$ associate to the right.  A \emph{debt}
is an expressions of the form $\ngg{A}$.  The variable $\Gamma$ ranges
over multisets containing $R$-expressions, and the variable $\Delta$
ranges over multisets that can contain both $E$-expressions and debts.
The variable $\Rules$ denotes some countable set of $R$-expressions.
The function $\bias{\cdot}$ is a \emph{bias assignment}: it maps
atomic expressions to the set $\{-2,-1,+1,+2\}$ (a similar bias
assignment was used in \cite{liang11apal}).  The atomic expression $A$
is in the \emph{linear realm} if $\bias{A}$ is $\pm 1$ and in the
\emph{classical realm} if $\bias{A}$ is $\pm 2$.  If $\bias{A}>0$ then
a debit rule can be used with $A$.  The variable $\Upsilon$ ranges
over finite multisets of atomic expressions in the classical realm,
and the variable $\Super$ ranges over atomic expressions in the
classical realm.


A \basic-proof is \emph{atomically closed} if all occurrences of the
\init rule in it involve only atomic expressions, \ie, they are of the
form $\twoseq{A}{A}$ for an atomic expression $A$.

\begin{prop}[Completeness of atomically closed \basic-proofs]
\label{prop:ps3 atomically closed}
If the sequent $\twoseq{}{\Delta}$ has a \basic-proof then it has an
atomically closed \basic-proof.
\end{prop}

\begin{proof}
A simple induction on the structure of $E$ shows that any occurrence
of $\twoseq{E}{E}$ in which $E$ is not an atomic expression can be
replaced by a proof that is atomically closed.
\end{proof}

The proofs below concerning the \basic proof system will implicitly
apply the structural rules for atomic expressions with bias
assignments of $\pm 2$.  In particular, the part of a context
composed of just such atomic expressions, usually denoted with the
$\Upsilon$ variable, will be treated additively even within
multiplicative rules.


\begin{prop}[Clip-admissibility for \basic-proofs]
\label{prop:clip ps3}
The following inference rule (a simpler version of Gentzen's cut
rule) is admissible in \basic.
\[
  \infer[\clip]{\twoseq{\Gamma_1,\Gamma_2}{\Delta_1,\Delta_2,\Upsilon}}
              {\twoseq{\Gamma_1}{\Delta_1,E,\Upsilon}\qquad
               \twoseq{\Gamma_2,E}{\Delta_2,\Upsilon}}
\]
\end{prop}

\begin{proof}
Consider the following \basic-proof with exactly one occurrence of the
\clip rule.
\[
  \infer[\clip]{\twoseq{\Gamma_1,\Gamma_2}{\Delta_1,\Delta_2,\Upsilon}}
               {\deduce{\twoseq{\Gamma_1}{\Delta_1,E,\Upsilon}}{\Xi_1}\qquad
                \deduce{\twoseq{\Gamma_2,E}{\Delta_2,\Upsilon}}{\Xi_2}}
\]
By Proposition~\ref{prop:ps3 atomically closed}, we can assume that
both $\Xi_1$ and $\Xi_2$ are atomically closed.  We proceed by
considering the structures of $\Xi_1$ and $\Xi_2$.  If either of these
proofs ends in a right rule for $\Delta_1$ or  $\Delta_2$, we can
permute those rule occurrences down.  Thus, we can assume that
$\Delta_1$ and $\Delta_2$ are multisets of atomic expressions.  Under
this assumption, we can also permute down any left rule that might
terminate $\Xi_2$.  In this case, we can assume that $\Gamma_2$ is
empty.  All that is left is showing how to permute the \clip rule up
into the left premise proof.

Consider the following instance of \clip.  Here, $E$ is either not 
atomic or it is atomic and $\bias{E}=\pm1$. 
\[
  \infer[\clip]{\twoseq{\Gamma_1,R\Mapsto E'}{\Delta_1,\Delta_2,\Upsilon}}
              {\infer{\twoseq{\Gamma_1,R\Mapsto E'}{E,\Delta_1,\Upsilon}}{
                \deduce{\twoseq{R}{\Upsilon}}{\Xi_1} &
                \deduce{\twoseq{\Gamma_1}{E,E',\Delta_1,\Upsilon}}{\Xi_2} &
}\qquad
               \deduce{\twoseq{E}{\Delta_2,\Upsilon}}{\Xi_3}}
\]
This instant can be rewritten to be
\[
  \infer{\twoseq{\Gamma_1,R\Mapsto E'}{\Delta_1,\Delta_2,\Upsilon}}
        {\deduce{\twoseq{R}{\Upsilon}}{\Xi_1} &
         \infer[\clip]{\twoseq{\Gamma_1}{E',\Delta_1,\Delta_2,\Upsilon}}{
                \deduce{\twoseq{\Gamma_1}{E, E', \Delta_1,\Upsilon}}{\Xi_2} &
                \deduce{\twoseq{E}{\Delta_2,\Upsilon}}{\Xi_3}}}
\]
In the case that $E$ is an atomic expression and $\bias{E}=\pm2$ then
the last inference rule of is either \init (in which case, \clip
is easily removed since $E\in\Upsilon$) or $\debit_2$ and, in that
case, $\Delta_2$ is empty (or a structural rule).  In this final case,
the proof above can be rewritten as
\[
  \infer{\twoseq{R\Mapsto E'}{\Delta_1,\Upsilon}}
        {\infer[\kern-4pt\clip]{\twoseq{R}{\Upsilon}}{
                      \deduce{\twoseq{R}{E,\Upsilon}}{\Xi_1}
                      &
                      \deduce{\twoseq{E}{\Upsilon}}{\Xi_3}}
         ~~
         \infer[\kern-4pt\clip.]{\twoseq{\Gamma_1}{E',\Delta_1,\Upsilon}}{
                \deduce{\twoseq{\Gamma_1}{E, E'\kern-2pt, \Delta_1,\Upsilon}}{\Xi_2} &
                \deduce{\twoseq{E}{\Upsilon}}{\Xi_3}}}
\]
The other cases regarding the structure of the $R$-expression in
$\Gamma$ are simple and direct.

The only remaining cases to consider is when $\Xi_1$ is a right rule
introducing $E$ and $\Xi_2$ is a left rule introducing $E$.
These cases are discussed below (remembering that $\Gamma_2$ is empty).

It is not possible for $E$ to be $\Zero$ since there is no such proof
$\Xi_2$.  If $E$ is $\One$, then $\Delta_2$ is empty and $\Xi_1$
replaces the clip rule.  If $E$ is $E_1\Plus E_2$ then we must have
\[
  \infer[\clip]{\twoseq{\Gamma_1}{\Delta_1,\Delta_2}}
              {\infer{\twoseq{\Gamma_1}{\Delta_1,E_1\Plus E_2}}{
                \deduce{\twoseq{\Gamma_1}{\Delta_1,E_1}}{\Xi_1'} &
                \deduce{\twoseq{\Gamma_1}{\Delta_1,E_2}}{\Xi_1''} &
}\qquad
               \infer{\twoseq{E_1\Plus E_2}{\Delta_2}}{
               \deduce{\twoseq{E_1}{\Delta_2}}{\Xi'_2}}}
\]
(where $\Delta_1$ and $\Delta_2$ contain only atomic expressions).
This instance of \clip can be replaced by the following instance of
\clip on smaller expressions.
\[
  \infer[\clip]{\twoseq{\Gamma_1}{\Delta_1,\Delta_2}}
              {\deduce{\twoseq{\Gamma_1}{\Delta_1,E_1}}{\Xi_1'} &
               \deduce{\twoseq{E_1}{\Delta_2}}{\Xi'_2}}
\]
The symmetric case is handled the same.
If $E$ is $E_1\Times E_2$ then we must have
\[
  \infer[\clip]{\twoseq{\Gamma_1}{\Delta_1,\Delta_2}}
              {\infer{\twoseq{\Gamma_1}{\Delta_1,E_1\Times E_2}}{
                \deduce{\twoseq{\Gamma_1}{\Delta_1,E_1,E_2}}{\Xi_1'}} &
               \infer{\twoseq{E_1\Times E_2}{\Delta',\Delta''_2}}{
               \deduce{\twoseq{E_1}{\Delta'_2}}{\Xi'_2} &
               \deduce{\twoseq{E_2}{\Delta''_2}}{\Xi''_2}}}
\]
(where $\Delta_1$ and $\Delta_2$ contain only atomic expressions).  This
instance of \clip can be replaced by the following instance of
\clip on smaller expressions. 
\[
  \infer[\clip]{\twoseq{\Gamma_1}{\Delta_1,\Delta'_2,\Delta'_2}}
              {\infer[\clip]{\twoseq{\Gamma_1}{\Delta_1,\Delta'_2,E_2}}{
                \deduce{\twoseq{\Gamma_1}{\Delta_1,E_1,E_2}}{\Xi_1'} &
                \deduce{\twoseq{E_1}{\Delta'_2}}{\Xi'_2}
} &
                \deduce{\twoseq{E_2}{\Delta''_2}}{\Xi''_2}}
\]

In general, one occurrence of clip can be replaced by two clips.
Standard induction arguments can now be used to complete this proof.
\end{proof}

\begin{prop}[Right rules are invertible]
\label{prop:invertible ps3}
The right rules are invertible.  In particular, if $E$ is
not atomic and the sequent $\twoseq{}{E,\Delta}$ is provable, then
there is a proof of this sequent in which the last inference rule is
an introduction rule for $E$.
\end{prop}

\begin{proof}
Let $\Xi$ be a proof of $\twoseq{}{E_1\Times E_2,\Delta}$.  Consider
\[
  \infer{\twoseq{\Gamma}{E_1\Times E_2,\Delta}}{
  \infer[\clip]{\twoseq{\Gamma}{E_1, E_2,\Delta}}{
      \deduce{\twoseq{}{E_1\Times E_2,\Delta}}{\Xi} &
      \infer{\twoseq{E_1\Times E_2}{E_1,E_2}}{
              \infer[\init]{\twoseq{E_1}{E_1}}{} & 
              \infer[\init]{\twoseq{E_2}{E_2}}{}}}}
\]
By Proposition~\ref{prop:clip ps3}, this proof with \clip can be
replaced by a proof without \clip: that proof ends in the
introduction of $E_1\Times E_2$.  The case where $E$ is $\One$ is
similar and simpler.  Let $\Xi$ be a proof of $\twoseq{}{E_1\Plus
  E_2,\Delta}$.  Consider
\[
  \infer{\twoseq{\Gamma}{\Delta,E_1\Plus E_2}}{
  \infer[\clip]{\twoseq{\Gamma}{\Delta,E_1}}{
      \deduce{\twoseq{}{E_1\Plus E_2,\Delta}}{\Xi} &
      \infer{\twoseq{E_1\Plus E_2}{E_1}}{
             \infer[\init]{\twoseq{E_1}{E_1}}{}}}
     & \infer[\clip]{\twoseq{\Gamma}{\Delta,E_2}}{
                    \hbox{similar}}}
\]
By Proposition~\ref{prop:clip ps3}, this proof with \clip can be
replaced by a proof without \clip.  The case where $E$ is $\Zero$ is
immediate.
\end{proof}

\begin{prop}[Clipping out debt]\label{prop:clip debt}
Let $C$ be an atomic expression.
If $\bias{C}=+1$, the following $\dclip_1$ rule is admissible.
\[
  \infer[\dclip_1]{\twoseq{}{\Delta_1,\Delta_2}}
                 {\twoseq{}{\Delta_1,C}\qquad
                  \twoseq{}{\Delta_2,\ngg{C}}}
\]
If $\bias{C}=+2$, the following $\dclip_2$ rule is admissible.
\[
  \infer[\dclip_2]{\twoseq{}{\Delta,\Upsilon}}
                 {\twoseq{}{\Delta,C}\qquad
                  \twoseq{}{\Upsilon,\ngg{C}}}
\]
\end{prop}

\begin{proof}
Consider the following \basic-proof with one occurrence of $\dclip_1$:
here, $\bias{C}=+1$. 
\[
  \infer[\dclip_1]{\twoseq{}{\Delta_1,\Delta_2}}
                 {\deduce{\twoseq{}{\Delta_1,C}}{\Xi_1}\qquad 
                  \deduce{\twoseq{}{\Delta_2,\ngg{C}}}{\Xi_2}}
\]
This proof can be rewritten as
\[
  \infer[\clip]{\twoseq{}{\Delta_1,\Delta_2}}
              {\deduce{\twoseq{}{\Delta_1,C}}{\Xi_1}\qquad 
               \infer[\debit_1]{\twoseq{C}{\Delta_2}}{
               \deduce{\twoseq{}{\Delta_2,\ngg{C}}}{\Xi_2}}}
\]
Apply Proposition~\ref{prop:clip ps3} to finish this case.  Consider
the following proof with one occurrence of $\dclip_2$: here,
$\bias{C}=+2$.
\[
  \infer[\dclip_2]{\twoseq{}{\Delta,\Upsilon}}
                 {\deduce{\twoseq{}{\Delta,C}}{\Xi_1}\qquad 
                  \deduce{\twoseq{}{\Upsilon,\ngg{C}}}{\Xi_2}}
\]
This proof can be rewritten as
\[
  \infer[\clip]{\twoseq{}{\Delta,\Upsilon}}
               {\deduce{\twoseq{}{\Delta,C}}{\Xi_1}\qquad 
                \infer[\debit_2]{\twoseq{C}{\Upsilon}}{
                \deduce{\twoseq{}{\Upsilon,\ngg{C}}}{\Xi_2}}}
\]
Apply Proposition~\ref{prop:clip ps3} to finish this case.
\end{proof}

If an atomic expression $A$ has a positive bias value, the \debit rule
allows turning an obligation into find $A$ in the current multiset
into a promise to pay that obligation later, possibly after additional
rewriting takes place.  The proposition theorem states that once a
complete proof is built, possibly using the \debit rules, it is
possible to reorganize that proof so that no \debit rules are used.

\begin{prop}[Completeness without debit]\label{prop:debit independent}
If the sequent $\twoseq{}{\Delta}$ has a \basic-proof, it has a proof
without the $\debit_1$ and $\debit_2$ rules.
\end{prop}

\begin{proof}
We systematically replace an occurrence of the $\debit_1$ inference
rule (above a \decide rule) with \init and $\dclip_1$
(below the \decide rule).  That is, we transform 
\[
  \infer[\decide]
        {\twoseq{}{\Delta_0,\Delta_1}}{
         \infer{\twoseq{\Gamma}{\Delta_0,\Delta_1}}
               {\strut\cdots\qquad & 
                \deduce{\vdots}{
                \infer[\debit_1]
                      {\twoseq{A}{\Delta_0}}
                      {\deduce{\twoseq{}{\ngg{A},\Delta_0}}{\strut\Xi}}}
                       & \cdots}}
\]
into the following proof containing $\dclip_1$.  Here, we 
replaced $\Delta_0$ with $A$ in some of the sequents and then used the
\clip rule to reintroduce the $\Delta_0$ expressions.
\[
  \infer[\dclip_1]{\twoseq{}{\Delta_0,\Delta_1}}{
  \infer[\decide]
        {\twoseq{}{A,\Delta_1}}{
         \infer{\twoseq{\Gamma}{A,\Delta_1}}
               {\strut\cdots\qquad & 
                \deduce{\vdots}{
                \infer[\init]
                      {\twoseq{A}{A}}
                      {}}
                       &
                       \cdots}}
         & \deduce{\twoseq{}{\ngg{A},\Delta_0}}{\strut\Xi}}
\]
We also can systematically replace an occurrence of the $\debit_2$
inference rule (above a \decide rule) with \init and $\dclip_2$
(below the \decide rule).  That is, we transform a \basic-proof of the form
\[
  \infer[\decide]
        {\twoseq{}{\Upsilon,\Delta}}{
         \infer{\twoseq{\Gamma}{\Upsilon,\Delta}}
               {\strut\cdots\qquad & 
                \deduce{\vdots}{
                \infer[\debit_1]
                      {\twoseq{A}{\Upsilon}}
                      {\deduce{\twoseq{}{\ngg{A},\Upsilon}}{\strut\Xi}}}
                       &
                       \cdots}}
\]
with the following proof with $\dclip_2$ below.  Here, we replaced
$\Upsilon$ with $A$ in some of the sequents and used the $\dclip_2$
rule to reintroduce the $\Upsilon$ expressions.
\[
  \infer[\dclip_2]{\twoseq{}{\Upsilon,\Delta}}{
  \infer[\decide]
        {\twoseq{}{A,\Delta}}{
         \infer{\twoseq{\Gamma}{A,\Delta}}
               {\strut\cdots\qquad & 
                \deduce{\vdots}{
                \infer[\init]
                      {\twoseq{A}{A}}
                      {}}
                       &
                       \cdots}}
         &
         \deduce{\twoseq{}{\ngg{A},\Upsilon}}{\strut\Xi}}
\]
Thus, we have replaced one occurrence of either $\debit_1$ or
$\debit_2$ with one occurrence of $\dclip_1$ or $\dclip_2$,
respectively.  Using Proposition~\ref{prop:clip debt}, we have a
clip-free proof with one fewer debit rules.  Note that clip
elimination does not introduce \debit when there is no \debit in the
original proof.
\end{proof}

A \basic-proof $\Xi$ is \emph{reduced} if every occurrence of
the \decide rule has a right-hand side containing only atomic
expressions or debts. 

The \emph{major premises} of the left rules are defined as follows.  Those
rules with only a single premise 
have that sole premise as their major premise.  Both premises of the
left-introduction rule for $\Times$ are major premises.
Finally, the left-most premise is the major premise for the
introduction rules for $\mapsto$ and $\Mapsto$.
Note that if the right-hand
side of the conclusion of a left rule occurrence contains only atomic
expressions, then this is true of the major premises of that rule
occurrence.

\begin{prop}[Completeness of reduced \basic-proofs]
\label{prop:reduced ps3}
If the sequent $\twoseq{}{\Delta}$ has a \basic-proof, it has a
reduced proof.
\end{prop}

\begin{proof}
An occurrence of a sequent in $\Xi$ is \emph{bad} if that sequent is
the conclusion of a left rule and a major premise of that rule is the
conclusion of a right-introduction rule.  Note that the right-hand
side of a bad sequent occurrence must contain a non-atomic expression.
The measure of a bad occurrence of a sequent is the height of its
subproof in $\Xi$.  The measure of the \basic-proof $\Xi$ is the
multiset of the measure of all bad sequents in $\Xi$.  We prove that
if the measure of $\Xi$ is not the empty multiset, then we can replace
$\Xi$ with another proof of the same end-sequent but with strictly
smaller multiset ordering.

Assume that the measure of $\Xi$ is non-empty.  Then there exists a
sequent with a bad occurrence in $\Xi$.  Pick one of these with
minimal height and assume that that sequent is of the form
$\twoseq{\Gamma}{\Delta}$.  As noted above, there must be a non-atomic
expression in $\Delta$.  Hence, the last left rule cannot be either
$\debit_2$ or the left-introduction rule for $\One$.  Thus, we only
need to consider six left rules (\decide, $\debit_1$, and one each for
$\Times, \Plus, \mapsto, \Mapsto$).  Since there are four
right introduction rules (one for each of $\One, \Times, \Zero,
\Plus$) then we have 24 possible combinations of rules that can yield
the bad occurrence $\twoseq{\Gamma}{\Delta}$.  If the upper rule is
the right-introduction of $\Zero$ or $\One$, then we can trivially
permute that rule down.  We illustrate a few more cases.  The
remaining ones are similar.
\[
\vcenter{
  \infer{\twoseq{R_1\times R_2}{E_1\Plus E_2,\Delta_1,\Delta_2}}{
  \infer{\twoseq{R_1}{E_1\Plus E_2,\Delta_1}}{
      \deduce{\twoseq{R_1}{E_1,\Delta_1}}{\Xi_1} &
      \deduce{\twoseq{R_1}{E_2,\Delta_1}}{\Xi_2}} &
  \deduce{\twoseq{R_2}{\Delta_3}}{\Xi_3}}}\quad
    \longrightarrow
\]
\[
  \infer{\twoseq{R_1\times R_2}{E_1\Plus E_2,\Delta_1,\Delta_2}}{
  \infer{\twoseq{R_1\times R_2}{E_1,\Delta_1,\Delta_2}}{
      \deduce{\twoseq{R_1}{E_1,\Delta_1}}{\Xi_1} &
      \deduce{\twoseq{R_2}{\Delta_2}}{\Xi_3}} &   
  \infer{\twoseq{R_1\times R_2}{E_2,\Delta_1,\Delta_2}}{
      \deduce{\twoseq{R_1}{E_2,\Delta_1}}{\Xi_2} &
      \deduce{\twoseq{R_2}{\Delta_2}}{\Xi_3}}}
\]
\[
\vcenter{
  \infer{\twoseq{R\mapsto E}{E_1\Times E_2,\Delta_1,\Delta_2}}{
      \infer{\twoseq{R}{E_1\Times E_2,\Delta_1}}
            {\deduce{\twoseq{R}{E_1,E_2,\Delta_1}}{\Xi_1}} &
      \deduce{\twoseq{}{E,\Delta_2}}{\Xi_2}}}\quad
    \longrightarrow
\]
\[
  \infer{\twoseq{R\mapsto E}{E_1\Times E_2,\Delta_1,\Delta_2}}{
  \infer{\twoseq{R\mapsto E}{E_1,E_2,\Delta_1,\Delta_2}}{
    \deduce{\twoseq{R}{E_1,E_2,\Delta_1}}{\Xi_1} &
    \deduce{\twoseq{}{E,\Delta_2}}{\Xi_2}}}
\]
(We have assumed that these sequents have a left-hand side with at
most two expressions: these cases are easily extended to the more
general case.)  Note that in the last pair of proofs, for example, the
bad occurrence of the sequent is moved up, but the sequent
$\twoseq{R\mapsto E}{E_1,E_2,\Delta_1}$ may be a bad occurrence in the
result: if that is the case, its measure has decreased.  In this way,
the measure decreases whenever we permute such rules.
\end{proof}

Since the \decide rule in \basic allows for deciding on $\Gamma$ with
multiple expressions, we say that \basic-proofs are, in general,
\emph{multi-decide} proofs.  A \basic-proof is a \emph{single-decide}
proof if every occurrence of the \decide rule in it decides on exactly
one expression.  While allowing multi-decide proofs was a convenience
for proving the clip-elimination result (Proposition~\ref{prop:clip
  ps3}), we maintain completeness by restricting to single-decide
proofs.

\begin{prop}[Completeness of single-decide proofs]
\label{prop:completeness single decide}
A \basic provable sequent has single-decide \basic-proof.
\end{prop}

\begin{proof}
In principle, deciding on multiple expressions can be done sequentially.
Since all left rules permute over each other, we can assume that the
left rules are done in a \emph{focused} manner: that is, the immediate
subexpressions of an $R$-expressions in major premises can be
introduced in the proof of that major premise (we also include the use
of \init or a \debit rule).    Schematically, we can then take
instances of the \decide rule of the form 
\[
  \infer[\decide]{\twoseq{}{\Delta}}{
  \infer{\twoseq{R,\Gamma}{\Delta}}{\strut\cdots & 
         \infer{\vdots}{\twoseq{\Gamma_i}{\Delta_i}} & 
         \cdots}}
\]
where $\Gamma$ is non-empty and $\Gamma_i$ is a sub-multiset of
$\Gamma$ and where $i\in\{1,\ldots,n\}$ for some positive $n$.  If
$\Gamma_i$ is non-empty, then we can transform this proof into
\[
  \infer[\decide]{\twoseq{}{\Delta}}{
  \infer{\twoseq{R}{\Delta}}
        {\deduce{\strut\cdots}{\deduce{\vdots}{\cdots}}\qquad & 
         \infer{\vdots}{\infer[\decide]{\twoseq{}{\Delta_i}}
                                     {\twoseq{\Gamma_i}{\Delta_i}}} & 
         \deduce{\strut\cdots}{\deduce{\vdots}{\cdots}}}}
\]
An inductive argument can be used to remove all \decide rules that
decide on more than one rule.
\end{proof}

\section{The two-phase inference system \twophase}
\label{sec:two phase}

The \basic inference system supports the basic features we
motivated in Section~\ref{sec:motivated} that should be present in an
inference system that supports the search for proofs.  At the same
time, \basic can be improved significantly to better support such
search.

In the previous section, we have taken steps in that direction
already.  The completeness of single-decide proofs means that we do
not need to consider selecting collections of rules at a time because
selecting them one at a time is just as complete.  Similarly, the
completeness of reduced proofs implies that the search for proofs can
be done by first doing all possible right rules, then selecting one
$R$-expression for the \decide rule, and then doing only left rules
along the major premises.

There are, however, still defects in the search for proofs since there
remains some non-determinism in the search for (reduced and
single-decide) \basic-proofs that can be removed.  For example,
$\twoseq{A}{A}$ can be proved using \init, but, if $\bias{A}=+1$, it
can also be proved using both $\debit_1$ and $\initd$.  Also, the
rules of contraction and weakening can be applied at almost any moment
during search.  

The two-phased proof system in \twophase, given in \fref{f3},
captures only reduced and single-decide proofs and
where these two non-deterministic choices are resolved.  There are two
kinds of sequents in \twophase, namely $\twoseq{}{\Delta}$ and
$\match{R}{\Atoms}$, where $\Delta$ is a multiset of $E$-expressions,
$R$ is an $R$-expression, and $\Atoms$ is a multiset of atomic
expressions and debts.  When comparing this proof system to \basic, there is a
clear separation on left and right rules. A sequent of the form
$\twoseq{}{\Atoms}$ is called a \emph{border} sequent.

Note that in \twophase, if $\bias{A}=+1$ and we encounter
$\match{A}{A}$, then only the $\debit_1$ and $\initd$ rules can be
used to prove it: the \initL rule is not available.  Also, the two
structural rules are built into this proof system using the schematic
variable $\Supers$ to denote a multiset of atomic expressions in the
classical realm: this is achieved by treating the part of the context
identified as $\Supers$ as additive even in multiplicative rules.

\begin{figure}
\[
  \infer{\twoseq{}{\Zero,\Delta}}{}
  \qquad
  \infer{\twoseq{}{E_1\Plus E_2,\Delta}}
        {\twoseq{}{E_1,\Delta} &
         \twoseq{}{E_2,\Delta}}
\]\[
  \infer{\twoseq{}{\One,\Delta}}
        {\twoseq{}{\Delta}}
  \qquad
  \infer{\twoseq{}{E_1\Times E_2,\Delta}}
        {\twoseq{}{E_1,E_2,\Delta}}
\]\[
  \infer[\decide,~ R\in\Rules]
        {\twoseq{}{\Atoms,\Supers}}
        {\match{R}{\Atoms,\Supers}}
  \qquad
  \infer{\match{R_1\Plus R_2}{\Atoms,\Supers}}
        {\match{R_i}{\Atoms,\Supers}}
\]\[
  \infer{\match{\One}{\Supers}}{}
  \qquad
  \infer{\match{R_1\Times R_2}{\Atoms_1, \Atoms_2,\Supers}}
        {\match{R_1}{\Atoms_1,\Supers} & 
         \match{R_2}{\Atoms_2,\Supers}}
\]\[
  \infer{\match{R\mapsto E}{\Atoms_1,\Atoms_2,\Supers}}
        {\match{R}{\Atoms_1,\Supers} & 
         \Ratch{}{E}{\Atoms_2,\Supers}}
  \qquad
  \infer{\match{R\Mapsto E}{\Atoms,\Supers}}
        {\match{R}{\Supers} & 
         \Ratch{}{E}{\Atoms,\Supers}}
\]\[
  \infer[\release\dag]{\Ratch{}{E}{\Atoms,\Supers}}
        {\twoseq{}{E,\Atoms,\Supers}}
  \qquad
  \infer[\initR]{\Ratch{}{A}{\ngg{A},\Supers}}{\bias{A}>0}
\]\[
  \infer[\initL]{\match{A}{A,\Supers}}{\bias{A}<0}
  \qquad
  \infer[\initd]{\twoseq{}{\ngg{A},A,\Supers}}{\quad\delta(A)>0}
\]\[
  \infer[\debit_1]{\match{A}{\Atoms,\Supers}}
                  {\twoseq{}{\ngg{A},\Atoms,\Supers}\quad\delta(A)=+1}
\]\[
  \infer[\debit_2]{\match{A}{\Supers}}
                  {\twoseq{}{\bar{A},\Supers}\quad\bias{A}=+2}
\]
\caption{The two-phase inference system \twophase.  The proviso $\dag$ for
  the \release rule states that $E$ is either not atomic or it is
  atomic and $\bias{E}<0$.}
\label{fig:f3}
\end{figure}

The proof of the following relative completeness theorem for \twophase
proofs follows from the completeness for reduced and
single-decide proofs (Propositions~\ref{prop:reduced ps3}
and~\ref{prop:completeness single decide}).

\begin{prop}
\label{prop:completeness3}
Let $\Delta$ be a multiset containing $E$-expressions and debits.
Then, $\vdash\Delta$ is provable in \basic if and only if
$\vdash\Delta$ is provable in \twophase.
\end{prop}

The proof system \twophase is a \emph{two-phase} proof system since
all of its inference rules can be organized into the following two
phases. 
A \emph{left phase} is a derivation composed of only left rules and
$\Downarrow$ sequents: this phase has a border sequent as its
conclusion, and its premises are the conclusion of either \release,
$\debit_1$, or $\debit_2$.  There are possibly many choices to make
during the construction of a left phase (the choice of $R\in\Rules$,
the choice of $i$ in the left rule for $\Plus$, and the choice of how
to split the side expressions among premises) and, as a result, this
phase encapsulates \emph{don't know non-determinism}.  A \emph{right
phase} is a derivation composed of only right rules: all of the
premises of this phase are border sequents, and its conclusion is
either the conclusion of the full proof or is the premise of either
\release, $\debit_1$, or $\debit_2$.  Note that there might be many
ways to build a right phase formally but they all relate their
conclusion to the same collection of premises.  In this sense, right
phases encapsulate \emph{don't care non-determinism}.

A \emph{synthetic rule} is composed of one left phase and zero or more
right phases, one for each premise of the left phase.  In particular,
the conclusion and all the premises of a synthetic rule are border
sequents.  We say that a synthetic rule is \emph{for} $R$ if the last
rule (necessarily a \decide rule) decides on $R$.

Note that the right rules and, hence, the right phase seen as a single
rule, is additive (see Section~\ref{sec:motivated}).  If there are no
atomic expressions with bias assignment $\pm 2$ then the left rules,
and the left phase seen as a single rule, are multiplicative.  If
atomic expressions have bias $\pm 2$ then these are treated additively
even in otherwise multiplicative rules.

The primary purpose of the \twophase proof system over the \basic
proof system is that the former is used to generate synthetic
inference rules from $R$ expressions.  In the next section, we provide
several illustrations of how $R$ expressions can be used to specify
various proof systems involving logical formulas.

\section{Applications of \basic and \twophase}

\subsection{Specifying Fibonacci numbers}

Denote by $\fib{n}$ the $n^{th}$ Fibonacci number and let $\Rules$ be
the union of $\{ \Fib{0}{0},~\Fib{1}{1}\}$ and the set
\[ 
  \{\Fib{n+2}{x+y}\mapsto\Fib{n+1}{x}\mapsto \Fib{n}{y}\kern
  -0.5pt\sep n,x,y\in\mathbb{N} \}.
\]
To determine the synthetic rules that can arise from $\Rules$,
consider the three cases for the value of $\bias{\Fib{\cdot}{\cdot}}$.

If $\bias{\Fib{\cdot}{\cdot}}<0$ then the synthetic rules are
\[
  \infer{\twoseq{}{\Fib{0}{0}}}{} \qquad
  \infer{\twoseq{}{\Fib{1}{1}}}{} \qquad
  \infer{\twoseq{}{\Fib{n+2}{x+y}}}{\twoseq{}{\Fib{n+1}{x}} &
    \twoseq{}{\Fib{n}{y}}}
\]
The sequent $\twoseq{}{\Fib{n}{\fib{n}}}$ is has a unique 
proof using these rules, and its size is exponential in $n$.

If $\bias{\Fib{\cdot}{\cdot}}=+1$: then the synthetic rules are
\[
  \infer{\twoseq{}{\Delta}}{\twoseq{}{\Delta,\ngg{\Fib{0}{0}}}}
  \quad\kern -2pt
  \infer{\twoseq{}{\Delta}}{\twoseq{}{\Delta,\ngg{\Fib{1}{1}}}} \quad
  \infer{\twoseq{}{\ngg{\Fib{n+1}{x}}, \ngg{\Fib{n}{y}},\Delta}}
        {\twoseq{}{\ngg{\Fib{n+2}{x+y}},\Delta}}
\]
The sequent $\twoseq{}{\Fib{n}{\fib{n}}}$ is provable and 
the sizes of such proofs are exponential in $n$.  While bottom-up
reasoning is taking place, contraction is not available on debts.  As
a result, there is no sharing of previous computations. 

Finally, if $\bias{\Fib{\cdot}{\cdot}}=+2$, then the synthetic rules
are the same as the previous case except that $\Delta$ must be
replaced with $\Upsilon$.
  \[
    \infer{\twoseq{}{\Upsilon}}{\twoseq{}{\Upsilon,\ngg{\Fib{0}{0}}}}
    \quad
    \infer{\twoseq{}{\Upsilon}}{\twoseq{}{\Upsilon,\ngg{\Fib{1}{1}}}}
    \quad \infer{\twoseq{}{\ngg{\Fib{n+1}{x}},
        \ngg{\Fib{n}{y}},\Upsilon}}
          {\twoseq{}{\ngg{\Fib{n+2}{x+y}},\Upsilon}}
  \]
The sequent $\twoseq{}{\Fib{n}{\fib{n}}}$ is provable only when $n\le 3$.

Another specification of the Fibonacci series uses a more deliberate
reuse strategy.  Let $n\ge 0$ and let $\Rules_n$ be the
set of rules that is the union of the singleton
$\{\Fib{n}{\fib{n}}\mapsto \Zero\}$ and all the rules of the form
\[
  \Fib{m+1}{x}\Times\Fib{m}{y}\mapsto\Fib{m+2}{x+y}\Times\Fib{m+1}{x}
\]
where $m,x,y$ are natural numbers.  In this case, the sequent
$\twoseq{}{\Fib{0}{0}\Times\Fib{1}{1}}$ is provable from $\Rules_n$
with a proof of size linear in $n$.

\subsection{Classical and intuitionistic logic}
\label{ssec:classical}

The main reason to introduce \framework, via the \basic and \twophase
proof systems, is to provide a specification framework for inference
rules.  When comparing different proof systems (e.g., a target and an
encoding of it), three \emph{levels of adequacy} naturally
arise~\cite{nigam10jar}.  The weakest level of adequacy is
\emph{relative completeness}, which considers only {\em provability}:
a formula has a proof in one system if it has a proof in the other
system.  A stronger level of adequacy is of \emph{full completeness of
proofs}: the proofs in one system are in one-to-one correspondence
with proofs in the other system.  If one uses the term ``derivation''
for possibly incomplete proofs (proofs that may have open premises),
then the strongest version of adequacy is that of \emph{full
completeness of derivations}, where every derivation (such as
inference rules themselves) are in one-to-one correspondence with
those in the other system.

Unless otherwise mentioned, the encodings of proof systems described
below will all be at the highest level of adequacy.  In particular,
one inference rule in a target proof system (say, a rule in natural
deduction) will correspond to a synthetic rule in \twophase.

\begin{figure}
$$\begin{array}{c@{\quad}r@{\ }l@{\qquad}c@{\quad}r@{\ }l} (\iimpLm)
    &\lft{A\iimp B} &\mapsto\ \rght{A} \Mapsto \lft{B}\\ (\iimpRm)
    &\rght{A\iimp B} &\mapsto\ \lft{A}\Times\rght{B}\\ (\wedgeLa)
    &\lft{A\wedge B} &\mapsto\ \lft{A}\\ (\wedgeRa) &\rght{A\wedge B}
    &\mapsto\ \rght{A} \Plus \rght{B}\\ (\wedgeLa) &\lft{A\wedge B}
    &\mapsto\ \lft{B}&\\ (\veeLa) &\lft{A\vee B}
    &\mapsto\ \lft{A}\Plus\lft{B}\\ (\veeRa) &\rght{A\vee B}
    &\mapsto\ \rght{A}\\ (\veeRa) &\rght{A\vee B} &\mapsto\ \rght B\\
(\ifalseLa) &\lft{\ifalse} &\mapsto\ \Zero\\ (\itrueRa) &\rght{\itrue}
    &\mapsto\ \Zero\\ (Id_1) &\lft{C}\Times\rght{C} & \\ (Id_2) &\One
    &\mapsto\rght{C}\Mapsto\lft{C}
\end{array}$$
  \caption{Rules used to specify classical and intuitionistic logic.
    The superscript $a$ and $m$ on the names associated to
    $R$-expressions identify that rule as either additive or
    multiplicative.}
\label{fig:ljdef}
\end{figure}

Recalling now the discussion in Section~\ref{sec:motivated} regarding
representing the state of the search for a proof as a collection of
sheets, these sheets are represented as multisets of atomic
expressions of the form $\lft{B}$ and $\rght{B}$, where $B$ denotes a
logical formula.  Here, the expression $\lft{B}$ tags $B$ as coming at
the top of a sheet while $\rght{B}$ tags $B$ as coming at the bottom
of a sheet.

The rules in \fref{ljdef} can be used to describe natural deduction in
intuitionistic logic \emph{and} the sequent calculus for both
intuitionistic and classical logic.  In all of these cases, classical
logic is captured using the polarities $\bias{\lft{\cdot}}=\pm 2$ and
$\bias{\rght{\cdot}}=\pm 2$.  In contrast, intuitionistic logic is
captured using the polarities $\bias{\lft{\cdot}}=\pm 2$ and
$\bias{\rght{\cdot}}=\pm 1$.  Here, we are considering only
propositional logic with the logical constants $\iimp$ (implication),
$\wedge$ (conjunction), $\vee$ (disjunction), $\itrue$ (truth), and
$\ifalse$ (false). (First-order quantification is addressed in
Section~\ref{ssec:quantication}.)

\subsection{Natural deduction for intuitionistic logic}
\label{ssec:nd}

If we set $\bias{\lft{\cdot}}=+2$ and $\bias{\rght{\cdot}}=-1$, then
the synthetic rules derived in \twophase for the rule expressions in
\fref{ljdef} describe natural deduction proofs in intuitionistic
logic.  To prove this claim, we take the rules in \fref{natded} as the
formal definition of natural deduction~\cite{sieg98sl}.

\begin{figure}
\[
\infer[{[\iimp E]}]{\Gamma \vdash B\downarrow} {\Gamma \vdash A\iimp
  B\downarrow & \Gamma \vdash A\uparrow} \qquad \infer[{[\iimp
      I]}]{\Gamma \vdash A\iimp B \uparrow} {\Gamma,A \vdash B
  \uparrow}
\]
\[
\infer[{[\wedge E]}]{\Gamma \vdash A\downarrow} {\Gamma \vdash A
  \wedge B \downarrow} \qquad \infer[{[\wedge E]}]{\Gamma \vdash
  B\downarrow} {\Gamma \vdash A \wedge B \downarrow}
\]\[
\infer[{[\wedge I]}]{\Gamma \vdash A\wedge B \uparrow} {\Gamma \vdash
  A \uparrow & \Gamma \vdash B \uparrow}
\]
\[
\infer[{[\textrm{I}]}]{\Gamma,A \vdash A\downarrow}{} \qquad
\infer[{[\textrm{M}]}]{\Gamma \vdash A \uparrow} {\Gamma \vdash A
  \downarrow} \qquad \infer[{[\textrm{S}]}]{\Gamma \vdash A
  \downarrow} {\Gamma \vdash A \uparrow}
\]\[
\infer[{[\itrue I ]}]{\Gamma \vdash \itrue \uparrow }{} \qquad
\infer[{[\ifalse E ]}]{\Gamma \vdash C \uparrow }{\Gamma \vdash
  \ifalse \downarrow}
\]
\caption{The rules for the $\iimp$, $\iforall$, and $\wedge$ fragment
  of intuitionistic natural deduction NJ.}
\label{fig:natded}
\[
\infer[{[\vee E]}]{\Gamma \vdash C~\mathord\uparrow (\downarrow)}
      {\Gamma \vdash A\vee B\downarrow\ \Gamma,A \vdash
        C~\mathord\uparrow (\downarrow)\ \Gamma,B \vdash
        C~\mathord\uparrow (\downarrow)}
\]\[
\infer[{[\vee I]}]{\Gamma \vdash A_1\vee A_2 \uparrow} {\Gamma \vdash
  A_i \uparrow}
\]
\caption{The rules for $\vee$ for intuitionistic natural deduction. In
  $[\vee L]$, $i\in\{1,2\}$.}
\label{fig:natded2}
\end{figure}

Let $\Gamma\cup\{C\}$ be a set of propositional formulas and assume
that all $\bias{\lft{\cdot}}=+2$ and $\bias{\rght{\cdot}}=-1$.
The two judgments in \fref{natded} will be encoded as follows.  The
up-arrow judgment $\Gamma\vdash C\uparrow$ is encoded using
$\twoseq{}{\lft{\Gamma},\rght{C}}$.  The down-arrow judgment
$\Gamma\vdash C\downarrow$ is encode using
$\twoseq{}{\lft{\Gamma},\ngg{\lft{C}}}$.

Consider, for example, the following derivation using
the $(\iimpLm)$ rule in \fref{ljdef}.
\[
  \infer[1]{\twoseq{}{\ngg{\lft{B}},\Upsilon}}{
  \infer{\match{(\lft{A\iimp B}\mapsto \rght{A})\Mapsto 
                        \lft{B}}{\ngg{\lft{B}},\Upsilon}}{
    \infer{\match{\lft{A\iimp B}\mapsto \rght{A}}{\Upsilon}}{
       \infer[2]{\match{\lft{A\iimp B}}{\Upsilon}}
             {\twoseq{}{\ngg{\lft{A\iimp B}},\Upsilon}} &
              \infer[3]{\Ratch{}{\rght{A}}{\Upsilon}}
                               {\twoseq{}{\rght{A},\Upsilon}}} &
  \infer[4]{\Ratch{}{\lft{B}}{\ngg{\lft{B}},\Upsilon}}{}}}
\]
This derivation uses the \twophase rules (1) $\decide$, (2)
$\debit_2$, (3) $\release$, and (4) $\initR$.  The associated
synthetic inference rule is thus
\[
  \infer{\twoseq{}{\ngg{\lft{B}},\Upsilon}}
        {\twoseq{}{\ngg{\lft{A\iimp B}},\Upsilon} &
         \twoseq{}{\rght{A},\Upsilon}}
\]
In this example, since $\Upsilon$ can only contain atomic expressions
of the form $\lft{\cdot}$, we can write $\lft{\Gamma}$ for $\Upsilon$.
Thus, we have correctly captured the $[\iimp E]$ inference rule in
\fref{natded}. 

Deciding on $(Id_1)$ and $(Id_2)$, respectively, yields
\[
  \infer{\twoseq{}{\rght{B},\Upsilon}}{
  \infer{\match{\lft{B}\Times\rght{B}}{\rght{B},\Upsilon}}{
     \infer[\debit_2]{\match{\lft{B}}{\Upsilon}}{\twoseq{}{\ngg{\lft{B}},\Upsilon}} &
     \infer[\initL]{\match{\rght{B}}{\rght{B},\Upsilon}}{}}}
\]\[
  \infer{\twoseq{}{\ngg{\lft{B}},\Upsilon}}{
  \infer{\match{\One\mapsto\rght{B}\Mapsto\lft{B}}{\ngg{\lft{B}},\Upsilon}}{
     \infer{\match{\One\mapsto\rght{B}}{\Upsilon}}{
            \infer{\match{\One}{\Upsilon}}{} &
            \infer[\release]{\Ratch{}{\rght{B}}{\Upsilon}}{\twoseq{}{\rght{B},{\Upsilon}}}} &
     \infer[\initL]{\Ratch{}{\lft{B}}{\ngg{\lft{B}},\Upsilon}}{}}}
\]
and these yield the two synthetic rules
\[
  \infer[\quad\hbox{and}\quad]
        {\twoseq{}{\rght{B},\Upsilon}}{\twoseq{}{\ngg{\lft{B}},\Upsilon}}
  \quad
  \infer[.]{\twoseq{}{\ngg{\lft{B}},\Upsilon}}{\twoseq{}{\rght{B},\Upsilon}}
\]
These rules encode the natural deduction rules $[M]$ and $[S]$ rules,
respectively.

Consider the synthetic rules using the $(\veeLa)$ rule in \fref{ljdef}
for a final example.
\[
  \infer{\match{\lft{A\vee B}\mapsto \lft{A}\Plus\lft{B}}{\Delta,\Upsilon}}{
     \infer[\debit_2]{\match{\lft{A\vee B}}{\Upsilon}}
                     {\twoseq{}{\ngg{\lft{A\vee B}},\Upsilon}}
     \infer[\release]{\Ratch{}{\lft{A}\Plus\lft{B}}{\Delta,\Upsilon}}{
     \infer{\twoseq{}{\lft{A}\Plus\lft{B},\Delta,\Upsilon}}{
            \twoseq{}{\lft{A},\Delta,\Upsilon} &
            \twoseq{}{\lft{B},\Delta,\Upsilon}}}}
\]
Note that $\Delta$ could be either $\ngg{\lft{C}}$ or $\rght{C}$ for
some formula $C$.  As a result, the left introduction for disjunction
can appear in either the $\downarrow$ or $\uparrow$ style judgments.
Thus, this synthetic inference rule faithfully captures the $(\vee E)$
inference rule in \fref{natded2}.

Let $\Gamma\njDash C\uparrow$ and $\Gamma\njDash C\downarrow$ denote,
respectively, the facts that $\Gamma\vdash C\uparrow$ and
$\Gamma\vdash C\downarrow$ are provable using the rules in
\fref{natded} and~\ref{fig:natded2}.  Let $\Rules_{nj}$ be
the rules in \fref{ljdef}.  The following proposition holds.

\begin{prop}\label{prop:nj encoded}
Let $\Gamma\cup\{C\}$ be a set of object-level formulas and assume
that all $\bias{\rght{\cdot}}=-1$ and $\bias{\lft{\cdot}}=+2$.  Then
$\Gamma\njDash C\uparrow$ if and only if
$\twoseq{}{\lft{\Gamma},\rght{C}}$ is provable using $\Rules_{nj}$.
and $\Gamma\njDash C\downarrow$ if and only if
$\twoseq{}\lft{\Gamma},\ngg{\lft{C}}$ is provable using $\Rules_{nj}$.
\end{prop}

If the disjunction $\vee$ is removed, then derivations are considered
\emph{normal} (also, \emph{cut free}) if they do not use switch rule
($[S]$ rule in \fref{natded}).  Thus, normal proofs can be encoded for
such formulas simply by removing $(Id_2)$ from consideration in
Proposition~\ref{prop:nj encoded}.  See \cite{nigam10jar} for a
similar result but where \twophase is replaced by a linear logic proof
system.

\subsection{Sequent calculi for classical and intuitionistic logic}
\label{ssec:sequents}

When the polarities attributed to $\lft{\cdot}$ and $\rght{\cdot}$ are
both negative, the synthetic rules based on the rules in \fref{ljdef}
encode sequent calculus proofs.  For an example, if we assign
$\bias{\lft{\cdot}}=-2$ and $\bias{\rght{\cdot}}=-1$, then the
implication left rule $(\iimpLm)$ yields the following synthetic
inference rule.
\[
  \infer{\twoseq{}{\lft{A\iimp B},\Delta,\Upsilon}}{
         \twoseq{}{\rght{A},\lft{A\iimp B},\Upsilon} &
         \twoseq{}{\lft{B},\lft{A\iimp B},\Delta,\Upsilon}}
\]
This synthetic inference rule encodes the sequent calculus rule
(assuming that $\Delta$ is the multiset consisting of one occurrence
of $\rght{C}$).
\[
  \infer{\twoseq{A\iimp B,\Gamma}{C}}
        {\twoseq{A\iimp B,\Gamma}{A} &
         \twoseq{A\iimp B, B,\Gamma}{C}}
\]

If we change the bias assignment so that $\bias{\lft{\cdot}}=-2$ and
$\bias{\rght{\cdot}}=-2$ and consider the same implication-left
inference rule, then the same development holds except that the
multiset $\Delta$ is empty since all atoms belong to the classical
realm: the schema variable $\Upsilon$ will hold
atoms of both the form $\lft{\cdot}$ and $\rght{\cdot}$.  As a result,
we get the derived inference rule
\[
  \infer[.]{\twoseq{A\iimp B,\Gamma}{\Psi}}
           {\twoseq{A\iimp B,\Gamma}{A,\Psi} &
            \twoseq{A\iimp B, B,\Gamma}{\Psi}}
\]

As with the natural deduction calculus, the $(Id_1)$ and $(Id_2)$
rules have special roles.  In particular, using \decide with them
yields the following.
\[
  \infer[\decide~Id_1]
        {\match{\lft{C}\Times\rght{C}}{\lft{C},\rght{C},\Upsilon}}{
         \infer[\initL]{\match{\lft{C} }{\lft{C}, \Upsilon}}{} &
         \infer[\initL]{\match{\rght{C}}{\rght{C},\Upsilon}}{}}
\]
\[
  \infer[\decide~Id_2]
        {\match{\One\Mapsto\rght{C}\mapsto\lft{C}}{\Delta,\Upsilon}}{
         \infer{\match{\One}{\Upsilon}}{} &
         \twoseq{}{\rght{C}, \Upsilon} &
         \twoseq{}{\lft{C},\Delta,\Upsilon}}
\]
These are the following synthetic rules
\[
  \infer{\twoseq{}{\lft{C},\rght{C},\Upsilon}}{\strut}
  \qquad
  \infer{\twoseq{}{\Delta,\Upsilon}}{
         \twoseq{}{\rght{C}, \Upsilon} &
         \twoseq{}{\lft{C},\Delta,\Upsilon}}
\]
In the intuitionistic setting, the variable $\Upsilon$ contains only
$\lft{\cdot}$ atomic expressions while $\Delta$ contains only a single
expression and that is of the form $\rght{\cdot}$.  Thus,
$(Id_1)$ and $(Id_2)$ encode the \init and \cut rules of sequent
calculus.  (This encoding works for both intuitionistic and classical
logic.)

Given this discussion, the following has a direct proof.  

\begin{prop}[Negative bias encodes sequent calculus]
\label{prop:sequent}
If $\bias{\lft{\cdot}}=-2$ and $\bias{\rght{\cdot}}=-1$ then the rules
in \fref{ljdef} encode a sequent calculus proof system
(similar to Gentzen's \LJ proof system) which is complete for
intuitionistic logic.  If, however, we change the bias assignment so
that $\bias{\rght{\cdot}}=-2$, then the rules in \fref{ljdef}
encode a sequent calculus proof system (similar to Gentzen's \LK proof
system) which is complete for classical logic.
\end{prop}

By using Propositions~\ref{prop:debit independent}, \ref{prop:nj
  encoded}, and \ref{prop:sequent}, we can conclude immediately that
if a formula has a natural deduction proof then it has a sequent
calculus proof, since the only difference between these two encodings
is the use of the debit rules.

It is worth noting that if $\bias{\cdot}$ is modified so that for some
atomic expressions $A$, the value of $\bias{A}$ changes from $-1$ to
$-2$, then proofs in \basic remain proofs.  Thus, it is immediate that
sequent provability in intuitionistic logic yields sequent provability
in classical logic.

It is possible to encode the sequent calculus for both classical and
intuitionistic logic at a more primitive level: that is, by using
only the linear realm.  In particular, consider the specification in
\fref{ljdefalt}.  If $\bias{\lft{\cdot}}=-1$ and
$\bias{\rght{\cdot}}=-1$, then these rules yield sequent calculus
proofs for intuitionistic logic.  Dropping the use of the classical
realm affected this specification in two ways.  First, we needed to
add explicit weakening and contraction rules for left formula (the
rules $(LW)$ and $(LC)$, respectively).  Second, in encoding the
implication-left rule and the cut rule, the occurrences of the
right-side formula must be explicitly addressed in the rule's
specification.  In order to capture classical sequent calculus, we can
modify the rules in \fref{ljdefalt} by replacing the left rule for
implication and the $(Id_2)$ rule with the rules
\[
\begin{array}{c@{\quad}r@{\ }l}
  & \lft{A\iimp B}&\mapsto\rght{A} \mapsto\lft{B}\\
  & \One          &\mapsto\rght{C}\mapsto\lft{C}\\
\end{array}
\] 
and by adding the following explicit rules for weakening and
contraction for right tagged formulas. 
\[
(RW)~\rght{B} \mapsto\ \One \qquad
(RC)~\rght{B} \mapsto\ \rght{B}\Times\rght{B}
\]



\begin{figure}
$$\begin{array}{c@{\quad}r@{\ }l}
 &\lft{A\iimp B}\Times\rght{C}  &\mapsto\ \rght{A}\mapsto\lft{B}\Times\rght{C}\\
 &\rght{A\iimp B}   &\mapsto\ \lft{A}\Times\rght{B}\\
 &\lft{A\wedge B}   &\mapsto\ \lft{A}\\
 &\rght{A\wedge B}  &\mapsto\ \rght{A} \Plus \rght{B}\\
 &\lft{A\wedge B}   &\mapsto\ \lft{B}\\
 &\lft{A\vee B}     &\mapsto\ \lft{A}\Plus\lft{B}\\
 &\rght{A\vee B}    &\mapsto\ \rght{A}\\
 &\rght{A\vee B}    &\mapsto\ \rght B\\
 &\lft{\ifalse}     &\mapsto\ \Zero\\
 &\rght{\itrue}     &\mapsto\ \Zero\\
(Id_1) &\lft{C}\Times\rght{C} & \\
(Id_2) &\rght{A}     &\mapsto\rght{C}\mapsto\lft{C}\Times\rght{A}\\
(LW) &\lft{B}        &\mapsto\ \One \\
(LC) &\lft{B}        &\mapsto\ \lft{B} \Times\lft{B} \\
\end{array}$$
  \caption{Some rewrite rules used to specify sequent calculus proofs
    in intuitionistic logic entirely in the linear realm.}
\label{fig:ljdefalt}
\end{figure}

\subsection{Alternative encodings of proof rules}  

\begin{figure}
$$\begin{array}{c@{\quad}r@{\ }l@{\qquad}c@{\quad}r@{\ }l}
(\iimpLa) &\lft{A\iimp B}    &\mapsto\ \rght{A}\\
(\iimpRa) &\rght{A\iimp B}   &\mapsto\ \lft{A}\\
(\iimpRa)  &\rght{A\iimp B}   &\mapsto\ \rght{B}\\
(\wedgeLm) &\lft{A\wedge B}   &\mapsto\ \lft{A}\Times\lft{B}\\
(\wedgeRm) &\rght{A\wedge B}  &\mapsto\ \rght{A} \mapsto \rght{B}\\
(\veeLm) &\lft{A\vee B}     &\mapsto\ \lft{A}\mapsto\lft{B}\\
(\veeRm) &\rght{A\vee B}    &\mapsto\ \rght{A}\Times\rght{B}\\
(\ifalseLm) &\lft{\ifalse}     &\mapsto\ \One\\
(\itrueRm) &\rght{\itrue}     &\mapsto\ \One
\end{array}$$
  \caption{Some alternative version of inference rules.}
\label{fig:ljdefx}
\end{figure}

\fref{ljdefx} contains alternative specifications of the introduction
rules for some propositional logic constants.  In particular,
while \fref{ljdef} provides multiplicative rules for implication and
additive rules for conjunction, disjunction, true, and false, in
\fref{ljdefx}, we find additive rules for implication and multiplicative
rules for conjunction, disjunction, true, and false.  As is well known,
the presence of the structural rules (of weakening and contraction)
allows some pairing of these rules to be inter-admissible.

If we switch from the additive rules for conjunction ($\wedgeLa$ and
$\wedgeRa$ in \fref{ljdef}) to the multiplicative rules ($\wedgeLm$
and $\wedgeRm$ in \fref{ljdefx}), then the conjunction elimination
rule of intuitionistic natural deduction can be computed as follows.
\[
  \infer{\match{\lft{A\wedge B}\mapsto \lft{A}\Times\lft{B}}{\Delta,\Upsilon}}{
     \infer[\debit_2]{\match{\lft{A\wedge B}}{\Upsilon}}
                     {\twoseq{}{\ngg{\lft{A\wedge B}},\Upsilon}} &
     \infer[\release]{\Ratch{}{\lft{A}\Times\lft{B}}{\Delta,\Upsilon}}{
     \infer{\twoseq{}{\lft{A}\Times\lft{B},\Delta,\Upsilon}}{
            \twoseq{}{\lft{A},\lft{B},\Delta,\Upsilon}}}}
\]
Since $\Delta$ could be either $\ngg{\lft{C}}$ or $\rght{C}$ for some
formula $C$, the left-introduction for conjunction can appear in
either the $\downarrow$ or $\uparrow$ style judgments.
\[
\infer{\Gamma \vdash C~\mathord\uparrow (\downarrow)}
      {\Gamma \vdash A\wedge B\downarrow&
       \Gamma,A, B \vdash C~\mathord\uparrow (\downarrow)}
\]
This natural deduction rule is an example of a \emph{generalized
elimination rule}~\cite{schroeder-heister84,plato01aml}.

\subsection{Free deduction in classical logic}
\label{ssec:fd}

\begin{figure}
\[
\begin{array}{c@{\quad}r@{\ }l@{\qquad}c@{\quad}r@{\ }l}
 &\One\mapsto &  \lft{A\wedge B}\mapsto \rght{A}\mapsto \rght{B}\\
 &\One\mapsto & \rght{A\wedge B}\mapsto \lft{A}\\
 &\One\mapsto & \rght{A\wedge B}\mapsto \lft{B}\\
 &\One\mapsto & \lft{A\vee B}\mapsto \rght{A}\\
 &\One\mapsto & \lft{A\vee B}\mapsto \rght{B}\\
 &\One\mapsto & \rght{A\vee B}\mapsto \lft{A}\mapsto\lft{B}\\
 &\One\mapsto & \lft{A\iimp B}\mapsto \lft{A}\\
 &\One\mapsto & \lft{A\iimp B}\mapsto \rght{B}\\
 &\One\mapsto & \rght{A\iimp B}\mapsto \rght{A}\mapsto\lft{B}\\
\end{array}
\]
  \caption{Specification of the free deduction for classical logic.}
\label{fig:fd}
\end{figure}

\begin{figure}
\[
\begin{array}{c@{\quad}r@{\ }l}
 & \rght{A\wedge B}&\Times\ \lft{A}\Times\ \lft{B}\\
 & \lft{A\wedge B}&\Times\ \rght{A}\\
 & \lft{A\wedge B}&\Times\ \rght{B}\\
 & \rght{A\vee B}&\Times\ \lft{A}\\
 & \rght{A\vee B}&\Times\ \lft{B}\\
 & \lft{A\vee B}&\Times\ \rght{A}\Times\rght{B}\\
 & \rght{A\iimp B}&\Times\ \rght{A}\\
 & \rght{A\iimp B}&\Times\ \lft{B}\\
 & \lft{A\iimp B}&\Times\ \lft{A}\Times\rght{B}\\
\end{array}
\]
\caption{Several simple expressions provable from \klassical.}
\label{fig:simple}
\end{figure}

Given the interpretation of $(Id_1)$ and $(Id_2)$ as the $[M]$ and
$[S]$ inference rules in natural deduction, it is tempting to consider
both $\lft{B}$ and $\ngg{\rght{B}}$ and both $\rght{B}$ and
$\ngg{\lft{B}}$ as equivalent in some sense.  Such possible
equivalences do not immediately apply to rules, however, since rule
expressions do not contain debt expressions.  It might be possible,
however, to link proofs using a rule of the form $\lft{A}\mapsto
\lft{B}$ with a proof using a rule of the form
$\One\mapsto\rght{A}\mapsto\lft{B}$.  We illustrate such
considerations in this section.

The Free Deduction proof system~\cite{parigot92rclp} can be
encoded as follows.  Let \FD be the set of rules that results from
taking the rules in \fref{fd} along with $(Id_1)$ 
the following variant of $(Id_2)$:
\[
(Id_3)\qquad \One \mapsto\rght{C}\mapsto\lft{C}.
\]
Also, let \klassical be composed of the rules in \fref{ljdef}.  When
using both of these sets of rules, we assume that
$\bias{\lft{\cdot}}=\bias{\rght{\cdot}}=-2$.  As we have seen, under
this bias assignment, the rules in \klassical encode a classical
sequent system.  It is a simple exercise to show that all of the
expressions in \fref{simple} are provable from
\klassical.  Also, note the strong similarities between the rules in
\fref{fd} and the expressions in \fref{simple}: by dropping the
$\One\mapsto$ prefix, changing the remaining occurrences of $\mapsto$
to $\Times$, and flipping the left and right tags, we can convert
rules in \fref{fd} to expressions in \fref{simple}.

It is easy to show that every use of a rule in \FD can be emulated by
deciding on an expression in \fref{simple}.  For example, the
synthetic rule that results from the first rule in \fref{fd} is
\[
  \infer{\twoseq{}{\Delta_1,\Delta_2,\Delta_3}}
        {\twoseq{}{\lft{A\wedge B},\Delta_1} &
         \twoseq{}{\rght{A},\Delta_2} &
         \twoseq{}{\rght{B},\Delta_3}}
\]
This inference rule can be modeled in \basic by deciding on the first
expression in \fref{simple} and using $(Id_1)$ three times (and
with shifting the polarity to
$\bias{\lft{\cdot}}=\bias{\rght{\cdot}}=+2$):
\[
  \infer{\twoseq{}{\Delta_1,\Delta_2,\Delta_3}}{
  \infer{\twoseq{\rght{A\wedge B}\Times\lft{A}\Times\lft{B}}
                {\Delta_1,\Delta_2,\Delta_3}}{
         \twoseq{}{\ngg{\rght{A\wedge B}},\Delta_1} &
         \twoseq{}{\ngg{\rght{A}},\Delta_2} &
         \twoseq{}{\ngg{\rght{B}},\Delta_3}}}
\]
By using decide on the rule $(Id_3)$ on all three premises above, we
can build a \basic derivation that flips the debt $\ngg{\rght{A\wedge
    B}}$ to the atomic expression $\lft{A\wedge B}$ (as in the $[S]$
inference rule in Section~\ref{ssec:nd}).  It is now a simple matter
to use the clip-elimination theorem to remove the intermediate lemmas
listed in \fref{simple} for direct \basic-proofs.  Once we have such
\basic-proofs, Theorem~\ref{prop:completeness3} can
provide an \twophase proof corresponding to classical sequent calculus
proof.




\subsection{Linear logic}
\label{ssec:linear}

\begin{figure}
$$\begin{array}{c@{\quad}r@{\ }l}
(\limp L)    &\lft{A\limp B}    &\mapsto\ \rght{A}\mapsto\lft{B}.\\
(\limp R)    &\rght{A\limp B}   &\mapsto\ \lft{A} \Times\rght{B}.\\
(\otimes L)  &\lft{A\otimes B}  &\mapsto\ \lft{A} \Times\lft{B}.\\
(\otimes R)  &\rght{A\otimes B} &\mapsto\ \rght{A}\mapsto\rght{B}.\\
(\with L_1)  &\lft{A\with B}    &\mapsto\ \lft{A}.\\
(\with R)    &\rght{A\with B}   &\mapsto\ \rght{A}\Plus\rght{B}.\\
(\with L_2)  &\lft{A\with B}    &\mapsto\ \lft{B}.\\
(\oplus R_1) &\rght{A\oplus B}  &\mapsto\ \rght{A}.\\
(\oplus L)   &\lft{A\oplus B}   &\mapsto\ \lft{A}\Plus\lft{B}.\\
(\oplus R_2) &\rght{A\oplus B}  &\mapsto\ \rght{B}.\\
(\lpar L)    &\lft{A\lpar B}    &\mapsto\ \lft{A} \mapsto \lft{B}.\\
(\lpar R)    &\rght{A\lpar B}   &\mapsto\ \rght{A}\Times\rght{B}.\\
(\one L)     &\lft{\one}\hbox to 0pt{.\hss}       &\\
(\one R)     &\rght{\one}       &\Mapsto\ \One.\\
(\bottom L)  &\lft{\bottom}     &\Mapsto\ \One.\\
(\bottom R)  &\rght{\bottom}\hbox to 0pt{.\hss}   &\\
(\zero L)    &\lft{\zero}       &\mapsto\ \Zero.\\
(\top R)     &\rght{\top}       &\mapsto\ \Zero.\\
(\bang L)    &\lft{\bang B}     &\mapsto\ \llft{B}.\\
(\bang R)    &\rght{\bang B}    &\mapsto\ \rght{B}\Mapsto\One.\\
(\quest L)   &\lft{\quest B}    &\mapsto\ \lft{B} \Mapsto\One.\\
(\quest R)   &\rght{\quest B}   &\mapsto\ \rrght{B}.\\
(\derL)      &\llft{B}          &\mapsto\ \lft{B}.\\
(\derR)      &\rrght{B}         &\mapsto\ \rght{B}.
\end{array}$$
\caption{Specification of linear logic.}
\label{fig:ll-def}
\end{figure}

\fref{ll-def} contains a specification for linear logic.  This
specification makes use of \emph{four} tags: $\lft{\cdot}$,
$\rght{\cdot}$, $\llft{\cdot}$, and $\rrght{\cdot}$.  Here,
$\lft{\cdot}$ and $\rght{\cdot}$ construct atomic expressions that
should be in the linear realm while $\llft{\cdot}$ and $\rrght{\cdot}$
construct atomic expressions that should be in the classical realm.  A
sequent calculus for linear logic arises when we use the bias
assignment $\bias{\lft{\cdot}}=\bias{\rght{\cdot}}=-1$ and
$\bias{\llft{\cdot}}=\bias{\rrght{\cdot}}=-2$.

\subsection{Quantification}
\label{ssec:quantication}

Some of the earliest work on logic frameworks (for example, using
$\lambda$Prolog \cite{felty88cade} and the dependently typed LF
\cite{harper93jacm,felty90cade}) provided elegant approaches to the
treatment of quantifiers in the specification of proof systems.  The
essence of these treatments of quantifiers is described via the notion
of \emph{binder mobility} \cite{miller19jar}, a concept we illustrate
briefly here.  We first extended the grammar of $E$ and $R$ formulas
to allow both $\quant x (E~x)$ and $\quant x (R~x)$, where $\quant x$
is a binder for $x$ over expressions and rules.  Next, we need to add
to sequents a place for expression-level binders to move.  To this
end, we attach a variable-binding context $\Sigma$ to all sequents.
Thus, sequents have the structure $\sig{}\twoseq{\Gamma}{\Delta}$ and
$\sig{}\match{\Gamma}{\Delta}$.  In both of these cases, $\Sigma$ is a
list of distinct variables, all with scope intended over the formulas
in the respective sequent.  We assume the usual notions of $\alpha$,
$\beta$, and $\eta$ conversion.  The following two rules can be added
to \basic to treat quantifiers.
\[
  \infer{\sig{}\twoseq{\Gamma}{\quant x E~x,\Delta}}
        {\sig{,x}\twoseq{\Gamma}{E~x,\Delta}}
\qquad
  \infer{\sig{}\twoseq{\Gamma,\quant x R~x}{\Delta}}
        {\sig{}\twoseq{\Gamma,R~t}{\Delta}\quad
         t\hbox{ is a $\Sigma$-term}}
\]
In the first rule, we assume that $x$ is not already bound by
$\Sigma$.  In that rule, the expression-level binder for $x$ in the
conclusion is moved to a sequent-level binder for $x$ in the
premise.  The proviso in the second inference rule means that the free
variables of the (first-order) term $t$ are all taken from $\Sigma$.

Finally, to illustrate how quantifiers can be used to specify rules,
we first explicitly quantify over schema variables in the
specification of rules.  For example, the rule $(\iimpLm)$ in
\fref{ljdef} should be written more explicitly as 
\[
  \quant A\quant B \lft{A\iimp B} \mapsto\ \rght{A} \Mapsto \lft{B}
\]
Adding universal and existential quantification to the intuitionistic
and classical logic of Section~\ref{ssec:classical} can be done using
the (closed) $R$-expressions
$$\begin{array}{r@{\quad}r@{\ }l@{\qquad}c@{\quad}r@{\ }l}
\quant B\quant t &\lft{\forall x. B x}   &\mapsto\ \lft{B t}\\
\quant B         &\rght{\forall x. B x}  &\mapsto\ \quant x\rght{B x}\\
\quant B         &\lft{\exists x. B x}   &\mapsto\ \quant x\lft{B x}\\
\quant B\quant t &\rght{\exists x. B x}  &\mapsto\ \rght{B t}\\
\end{array}$$

\section{Related work}
\label{sec:related}

The two-phase proof system \twophase resembles \emph{uniform proofs},
which were used to describe logic programming as the search for proofs
in a two-phase proof system that alternated between a
\emph{goal-reduction} phase and a \emph{backchaining} phase
\cite{miller91apal}.  Andreoli's focused proof
system~\cite{andreoli92jlc} for Girard's linear
logic~\cite{girard87tcs} also inspired design aspects of
\framework. The closest related work, however, is the following
collection of papers that have used linear logic as a logical
framework for specifying proof systems.  The author showed how a
version of linear logic based on the negative connectives can
be used to specify sequent calculus and natural
deduction proof systems~\cite{miller96tcs}.  Nigam, Pimentel, and
others significantly extended such specifications, especially once
subexponentials were added to linear logic \cite{pimentel01phd},
\cite{nigam09phd}, \cite{nigam09ppdp}, \cite{nigam10jar},
\cite{nigam10lsfa}, \cite{miller13tcs}, \cite{nigam14jlc}.
Implementations and formal results surrounding such linear logic
specifications have also been built \cite{reis12tatu,olarte23jlamp}. A
design goal for \framework was to use it to replace linear logic as
the framework while attempting to find the fewest features of linear
logic that made it successful for specifying proof systems.

\section{Conclusion}

The state of the search for proofs in classical and intuitionistic
logic can be viewed as a collection of sheets of paper, each
containing assumptions and a conclusion: such sheets denote a gap in
the proof to be completed. An inference rule is encoded in reverse as
a rule for rewriting a sheet into 0 or more other sheets.  \framework
starts with this simple perspective of inference and formalizes
inference as the rewriting of collections of multisets of tagged
formulas. In doing so, the multiplicative and additive structures
behind logical inference are treated as primitive. This framework also
uses a bias assignment for tagged formulas that captures the notions
of linear and classical realm and of debt. We have also illustrated
how \framework specifications of inference rules can be used to
represent a range of known proof systems modularly. We demonstrated
this modularity by showing that one set of rewrite rules can account
for sequent calculus and natural deduction proofs in classical and
intuitionistic logic.

\bibliographystyle{IEEEtran}

\end{document}


%% file: macros.tex
\newcommand{\ie}{i.e.}
\newcommand{\Ascr}{\mathcal{A}}

\newcommand{\Cscr}{\mathcal{C}}

\newcommand{\Rscr}{\mathcal{R}}

\newcommand{\ProofSystem}[1]{{\rm\textbf{#1}}}
\newcommand{\framework}{\ProofSystem{PSF}\xspace}
\newcommand{\basic}{$\hbox{\ProofSystem{B}}$\xspace}
\newcommand{\twophase}{$\hbox{\ProofSystem{F}}$\xspace}

\newcommand{\LJ  }{\hbox{\ProofSystem{LJ}}\xspace}

\newcommand{\LK  }{\hbox{\ProofSystem{LK}}\xspace}

\newcommand{\limp}{\multimap}

\newcommand{\isemp}[1]{\ifthenelse{\isempty{#1}}{\cdot}{#1}}

\newcommand{\twoseq}[2]{#1\vdash #2}
\newcommand{\match}[2]{\mathop{\Downarrow} #1\vdash #2}
\newcommand{\Ratch}[3]{#1\vdash #2\mathbin{\Downarrow} #3}
\newcommand{\sep}{\mathbin{|}}

\newcommand{\contrR }{\hbox{\textsl{cR}}\xspace}

\newcommand{\weakR  }{\hbox{\textsl{wR}}\xspace}

\newcommand{\derR  }{\hbox{\textsl{derR}}\xspace}
\newcommand{\derL  }{\hbox{\textsl{derL}}\xspace}
\newcommand{\decide }{\hbox{\textsl{decide}}\xspace}

\newcommand{\release}{\hbox{\textsl{release}}\xspace}

\newcommand{\init    }{\hbox{\textsl{init}}\xspace}
\newcommand{\initd   }{\hbox{\textsl{iou}}\xspace}
\newcommand{\initL   }{\hbox{\textsl{initL}}\xspace}
\newcommand{\initR  }{\hbox{\textsl{initR}}\xspace}
\newcommand{\clip   }{\hbox{\textsl{clip}}\xspace}
\newcommand{\dclip  }{\hbox{\textsl{dclip}}\xspace}
\newcommand{\cut    }{\hbox{\textsl{cut}}\xspace}

\newcommand{\debit  }{\hbox{\textsl{debit}}\xspace}

\newcommand{\Rules  }{\ensuremath{\Rscr}\xspace}
\newcommand{\FD     }{\ensuremath{\mathcal{FD}}\xspace}
\newcommand{\klassical}{\ensuremath{\Cscr}\xspace}

\newcommand{\Zero}{\textbf{0}}
\newcommand{\One}{\textbf{1}}
\newcommand{\Plus}{\mathbin{\boldsymbol{+}}}
\newcommand{\Times}{\mathbin{\boldsymbol{\times}}}
\newcommand{\Quant}{\boldsymbol{Q}}
\newcommand{\quant}[1]{\Quant\,#1.}
\newcommand{\sig}[1]{\Sigma#1\,\colon}
\newcommand{\rrght}[1]{\lceil\kern -1pt\lceil#1\rceil\kern -1pt\rceil}
\newcommand{\llft}[1]{\lfloor\kern -1pt\lfloor#1\rfloor\kern -1pt\rfloor}
\newcommand{\rght}[1]{\lceil#1\rceil}
\newcommand{\lft}[1]{\lfloor#1\rfloor}

\newcommand{\iimp}{\supset}
\newcommand{\ifalse}{\mathord{\bottom}}
\newcommand{\itrue}{\top}
\newcommand{\iforall}{\forall_{i}}

\newcommand{\zero}{0}
\newcommand{\one}{1}
\newcommand{\bottom}{\perp}
\newcommand{\bang}{\mathop{!}}
\newcommand{\quest}{\mathord{?}}
\newcommand{\tensor}\otimes

\newcommand{\lpar}{\bindnasrepma}  

\newcommand{\ngg}[1]{\overline{#1}}

\newcommand{\bias}[1]{\delta(#1)}
\newcommand{\Atoms}{{ \Ascr}}
\newcommand{\Super}{S}

\newcommand{\wedgeLm}{\mathord{\wedge}L^m}
\newcommand{\wedgeRm}{\mathord{\wedge}R^m}
\newcommand{\wedgeLa}{\mathord{\wedge}L^a}
\newcommand{\wedgeRa}{\mathord{\wedge}R^a}
\newcommand{\veeLm}{\mathord{\vee}L^m}
\newcommand{\veeRm}{\mathord{\vee}R^m}
\newcommand{\veeLa}{\mathord{\vee}L^a}
\newcommand{\veeRa}{\mathord{\vee}R^a}
\newcommand{\iimpLm}{\mathord{\iimp}L^m}
\newcommand{\iimpRm}{\mathord{\iimp}R^m}
\newcommand{\iimpLa}{\mathord{\iimp}L^a}
\newcommand{\iimpRa}{\mathord{\iimp}R^a}
\newcommand{\ifalseLa}{\mathord{\ifalse}L^a}
\newcommand{\itrueRa}{\mathord{\itrue}R^a}
\newcommand{\ifalseLm}{\mathord{\ifalse}L^m}
\newcommand{\itrueRm}{\mathord{\itrue}R^m}

\newcommand{\mathsl}[1]{#1}
\newcommand\njDash{\vdash_{\mathsl{nj}}}

\newcommand{\Supers}{\Upsilon}

\newcommand{\fib}[1]{f(#1)}
\newcommand{\Fib}[2]{F(#1,#2)}
\newcommand{\fref}[1]{Fig.~\ref{fig:#1}}